\documentclass[runningheads,a4paper]{llncs}

\usepackage{amssymb}
\usepackage{amsmath}
\usepackage{graphicx}
\usepackage{tikz}
\usepackage{pgfplots}
\usepackage{framed}
\usepackage{listings}
\usepackage{algorithm,algpseudocode}
\usepackage{microtype}
\usepackage{epstopdf}
\usepackage{cite}
\usepackage{pifont}

\usepackage[firstpage]{draftwatermark}\SetWatermarkText{\hspace*{5.2in}\raisebox{5in}{\includegraphics[scale=0.1]{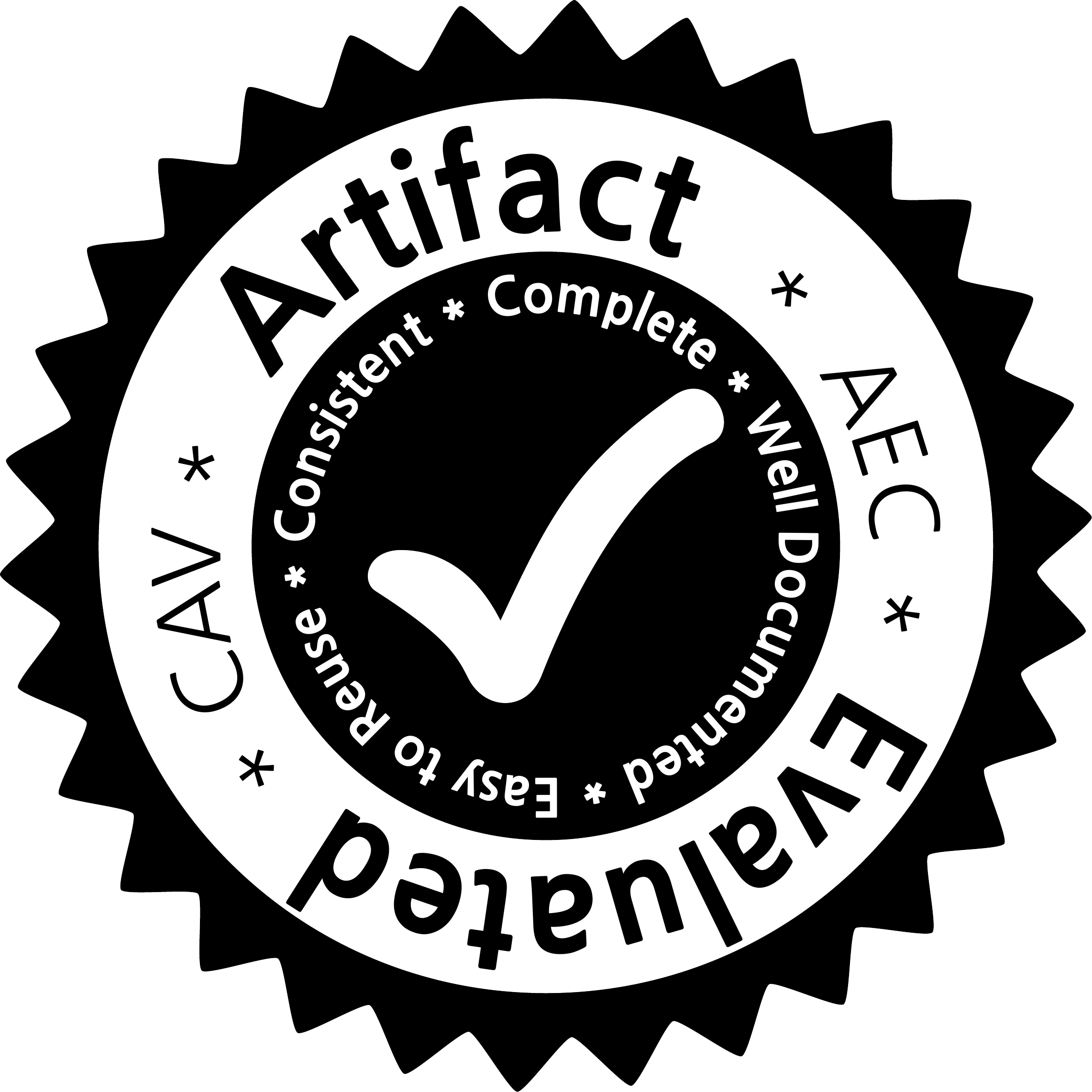}}}\SetWatermarkAngle{0}

\usepackage{todonotes}
\usetikzlibrary{positioning, automata, shapes.arrows, calc, shapes, arrows}
\usetikzlibrary{patterns}

\newcommand{\mat}[1]{\boldsymbol{#1}}
\newcommand{\xmark}{\ding{55}}

\usepackage{url}

\begin{document}

\newcommand\tool{{\sf DSSynth}\xspace}

\mainmatter  

\title{Automated Formal Synthesis of Digital Controllers for State-Space
Physical Plants\thanks{Supported by EPSRC grant EP/J012564/1, ERC project
280053 (CPROVER) and the H2020 FET OPEN 712689 SC$^2$.}}

\titlerunning{Automated Formal Synthesis of Digital Controllers}


%
%
%


%
%


\author{%
Alessandro Abate\inst{1} \and
Iury Bessa\inst{2} \and
Dario Cattaruzza\inst{1} \and
Lucas Cordeiro\inst{1,2} \and
\mbox{Cristina David}\inst{1} \and
Pascal Kesseli\inst{1} \and
Daniel Kroening\inst{1} \and
Elizabeth Polgreen\inst{1}}

\authorrunning{Abate et al.}

\institute{University of Oxford, UK
\and
Federal University of Amazonas, Manaus, Brazil}

\maketitle

\begin{abstract}
We present a sound and automated approach to synthesize safe digital
feedback controllers for physical plants represented as linear,
time-invariant models.  Models are given as dynamical equations with inputs,
evolving over a continuous state space and accounting for errors due to the
digitization of signals by the controller.  
Our counterexample guided inductive synthesis (CEGIS) approach has two phases: We synthesize a static feedback controller that stabilizes the system but that may not be safe for all initial conditions. Safety is then verified either via BMC or abstract acceleration; if the verification step fails, a counterexample is provided to the synthesis engine and the process iterates until a safe controller is obtained.
We demonstrate the practical value of this approach by automatically synthesizing safe controllers for intricate
physical plant models from the digital control literature. 
%
\end{abstract}

\section{Introduction}

Linear Time Invariant (LTI) models represent a broad class of dynamical
systems with significant impact in numerous application areas such as life
sciences, robotics, and engineering~\cite{astrom1997computer,Franklin15}. 
The synthesis of controllers for LTI models is well understood, however the
use of digital control architectures adds new challenges due to the effects
of finite-precision arithmetic, time discretization, and quantization noise,
which is typically introduced by Analogue-to-Digital (ADC) and
Digital-to-Analogue (DAC) conversion.  While research on digital control is
well developed~\cite{astrom1997computer}, automated and sound control
synthesis is challenging, particularly when the synthesis objective goes
beyond classical stability.  There are recent methods for verifying
reachability properties of a given controller~\cite{FLD+11}. 
However, these methods have not been generalized to control synthesis.  
Note that a synthesis algorithm that guarantees stability does not
ensure safety: the system might transitively visit an unsafe state resulting in unrecoverable failure.

We propose a novel algorithm for the synthesis of control algorithms for LTI
models that are guaranteed to be 
safe, considering both the continuous dynamics of the plant and the finite-precision discrete dynamics
of the controller, as well as the hybrid elements that connect them. 
%
%
We account for the presence of errors originating from a number of sources: 
quantisation errors in ADC and DAC, 
representation errors (from the discretization introduced by finite-precision arithmetic), 
and roundoff and saturation errors in the verification process 
(from finite-precision operations).  
Due to the complexity of such systems, we focus on linear models with known
implementation features (e.g., number of bits, fixed-point arithmetic). 
We~expect a safety requirement given as a reachability property.  Safety
requirements are frequently overlooked in conventional feedback control
synthesis, but play an important role in systems engineering.

We give two alternative approaches for synthesizing digital controllers for
state-space physical plants, both based on CounterExample Guided Inductive
Synthesis (CEGIS)~\cite{DBLP:conf/asplos/Solar-LezamaTBSS06}.  We prove
their soundness by quantifying errors caused by digitization and
quantization effects that arise when 
the digital controller
interacts with the continuous plant.

{\em The first approach} uses a na\"ive technique that starts by devising a
digital controller that stabilizes the system while remaining safe for a
pre-selected time horizon and a single initial state; then, it verifies
unbounded safety by unfolding the dynamics of the system, considering the
full hyper-cube of initial states, and checking a {\em completeness
threshold}~\cite{DBLP:conf/vmcai/KroeningS03}, i.e., the number of
iterations required to sufficiently unwind the closed-loop state-space
system such that the boundaries are not violated for any larger number of
iterations.  As~it requires unfolding up to the completeness threshold, this
approach is computationally expensive.

Instead of unfolding the dynamics, {\em the second approach}
employs {\em abstract acceleration}~\cite{cattaruzza2015unbounded} to
evaluate all possible progressions of the system simultaneously. 
Additionally, the second approach uses {\em abstraction refinement},
enabling us to always start with a very simple description regardless of the
dynamics complexity, and only expand to more complex models
when a solution cannot be found.

We provide experimental results showing that both approaches are able to
efficiently synthesize 
safe controllers for a set of intricate
physical plant models taken from the digital control literature: the median
run-time for our benchmark set is $7.9$ seconds, and most controllers can be
synthesized in less than $17.2$ seconds.  We further show that, in a direct
comparison, the abstraction-based approach (i.e., the second approach) lowers the median
run-time of our benchmarks by a factor of seven over the first approach
based on the unfolding of the dynamics.

%

\subsection*{Contributions} 

\begin{enumerate}
\item  We compute state-feedback controllers that 
  guarantee a given safety property.  Existing methods for
  controller synthesis rely on transfer function representations, which are
  inadequate to prove safety requirements.
\item We provide two novel algorithms: the first, na\"ive one, relies on
  an unfolding of the dynamics up to a completeness threshold, while the
  second one is abstraction-based and leverages abstraction refinement and
  acceleration to improve scalability while retaining soundness.  Both
  approaches provide sound synthesis of state-feedback systems and consider
  the various sources of imprecision in the implementation of the control
  algorithm and in the modeling of the plant.
\item We develop a model for different sources of quantization errors and
  their effect on reachability properties.  We give bounds that ensure the
  safety of our controllers in a hybrid continuous-digital domain.
\end{enumerate}

\section{Related Work}
\label{sec:relw}

\paragraph{CEGIS -}

Program synthesis is the problem of computing correct-by-design programs
from high-level specifications. Algorithms for this problem have made
substantial progress in recent years, for instance~\cite{itzhaky2010simple} to inductively synthesize invariants for the 
generation of desired programs.

Program synthesizers are an ideal fit for the synthesis of digital controllers, since
the semantics of programs capture the effects of finite-precision arithmetic
precisely.  In~\cite{DBLP:conf/cdc/RavanbakhshS15}, the authors use CEGIS
for the synthesis of switching controllers for stabilizing continuous-time
plants with polynomial dynamics.  The work extends to affine systems, but is
limited by the capacity of the state-of-the-art SMT solvers for solving
linear arithmetic.  Since this approach uses switching models instead of
linear dynamics for the digital controller, it avoids problems related to
finite precision arithmetic, but potentially suffers from state-space
explosion.  Moreover, in \cite{DBLP:conf/emsoft/RavanbakhshS16} the same
authors use a CEGIS-based approach for synthesizing continuous-time
switching controllers that guarantee \emph{reach-while-stay} properties of
closed-loop systems, i.e., properties that specify a set of goal states and
safe states (constrained reachability).  This solution is based on
synthesizing control Lyapunov functions for switched systems that yield
switching controllers with a guaranteed minimum dwell time in each mode. 
However, both approaches are unsuitable for the kind of control we seek to
synthesize.

The work in~\cite{hscc-paper} synthesizes stabilizing
controllers for continuous plants given as transfer functions by exploiting
bit-accurate verification of software-implemented digital
controllers~\cite{
Bessa16}.  While this work also uses CEGIS,
the approach is restricted to digital controllers for stable closed-loop
systems given as transfer function models: 
this results in  a static check on their coefficients.  
By contrast, in the current paper we consider a state-space representation of the physical system, 
which requires ensuring the specification over actual traces of the model, 
alongside the numerical soundness required by the effects of discretisation and finite-precision errors.   
A state-space model has known advantages over the transfer function
representation~\cite{Franklin15}: it naturally generalizes to multivariate systems
(i.e., with multiple inputs and outputs); 
and it allows synthesis of control systems with guarantees on the internal dynamics, e.g.,
to synthesize controllers that make the closed-loop system \emph{safe}.  Our
work focuses on the \emph{safety} of internal states, which is usually
overlooked in the literature.  Moreover, our work integrates an
abstraction/refinement (CEGAR) step inside the main CEGIS loop.

The tool Pessoa~\cite{mazo2010pessoa} synthesizes correct-by-design embedded
control software in a Matlab toolbox.  It is based on the abstraction of a
physical system to an equivalent finite-state machine and on the computation
of reachability properties thereon. 
Based on this safety specification, \mbox{Pessoa} can synthesize embedded controller
software for a range of properties.  The embedded controller software can be
more complicated than the state-feedback control we synthesize, and the
properties available cover more detail. 
%
%
However, relying on state-space discretization \mbox{Pessoa} is likely to incur in scalability limitations. 
Along this research line, \cite{Anta2010,liu16} studies the synthesis of digital controllers for continuous dynamics, 
and \cite{zamani2014} extends the approach to the recent setup of Network Control Systems. 


\paragraph{Discretization Effects -}

The classical approach to control synthesis has often disregarded digitalization effects, 
whereas more recently modern techniques have focused on
different aspects of discretization, including delayed
response~\cite{Duggirala2015} and finite word length (FWL) semantics, 
with the goal either to verify (e.g.,~\cite{daes20161}) or to optimize
(e.g.,~\cite{oudjida2014design}) given implementations. 

There are two different problems that arise from FWL semantics.  The first
is the error in the dynamics caused by the inability to represent the exact
state of the physical system, while the second relates to rounding and saturation errors
during computation.  In~\cite{fialho1994stability}, a stability measure
based on the error of the digital dynamics ensures that the deviation
introduced by FWL does not make the digital system unstable.  A~more recent
approach~\cite{DBLP:journals/automatica/WuLCC09} uses $\mu$-calculus to
directly model the digital controller so that the selected parameters are
stable by design.  The analyses in~\cite{DBLP:conf/hybrid/RouxJG15,
DBLP:conf/hybrid/WangGRJF16} rely on an invariant computation on the
discrete system dynamics using Semi-Definite Programming (SDP).  While the
former uses bounded-input and bounded-output (BIBO) properties to determine
stability, the latter uses Lyapunov-based quadratic invariants.  In both
cases, the SDP solver uses floating-point arithmetic and soundness is
checked by bounding the error.  An alternative is~\cite{park2016scalable},
where the verification of given control code is performed against a known
model by extracting an LTI model of the code by symbolic execution:  
to account for rounding errors, an upper bound is introduced in the
verification phase.  The work in \cite{picasso2003stabilization}
introduces invariant sets as a mechanism
to bound the quantization error effect on stabilization as an invariant set
that always converges toward the controllable set.  Similarly,
\cite{liberzon2003hybrid} evaluates the quantization error dynamics
and bounds its trajectory to a known region over a finite time period. 
This technique works for both linear and non-linear systems.


\section{Preliminaries}
\label{sec:preliminaries}

\subsection{State-space representation of physical systems} 
\label{ssec:ssrepresentation}

We consider models of physical plants expressed as ordinary differential
equations (ODEs), which we assume to be controllable and under full state
information (i.e., we have access to all the model variables):
\begin{align}
\label{eq:ode}
\dot{x}(t) = Ax(t)+ B u(t), \quad x \in \mathbb{R}^{n}, u \in \mathbb{R}^m, A \in \mathbb{R}^{n \times n}, B \in \mathbb{R}^{n \times m}, 
\end{align}
where $t \in \mathbb R_0^+$, where $A$ and $B$ are matrices that fully
specify the continuous plant, and with initial states set as $x(0)$.  While
ideally we intend to work on the continuous-time plant, in this work
Eq.~\eqref{eq:ode} is soundly discretized in time~\cite{fadali} into
%
%
\begin{align}
\label{eq:plant}
x_{k+1} = A_d x_k+ B_d u_k
\end{align} 
where $k \in \mathbb N$ and $x_{0}=x(0)$ is the initial state. 
$A_d$ and $B_d$ denote the matrices that describe the discretized plant dynamics, whereas
$A$ and $B$ denote the continuous plant dynamics.  
We synthesize for requirements over this discrete-time domain. 
Later, we will address the issue of variable quantization, 
as introduced by the ADC/DAC conversion blocks (Fig.~\ref{fig:digitalsystem}).

\begin{figure*}[htb]
\centering

\tikzset{add/.style n args={4}{
    minimum width=6mm,
    path picture={
        \draw[circle] 
            (path picture bounding box.south east) -- (path picture bounding box.north west)
            (path picture bounding box.south west) -- (path picture bounding box.north east);
        \node[draw=none] at ($(path picture bounding box.south)+(0,0.13)$)     {\small #1};
        \node[draw=none] at ($(path picture bounding box.west)+(0.13,0)$)      {\small #2};
        \node[draw=none] at ($(path picture bounding box.north)+(0,-0.13)$)    {\small #3};
        \node[draw=none] at ($(path picture bounding box.east)+(-0.13,0)$)     {\small #4};
        }
    }
 }

\resizebox{1.0\textwidth}{!}{
 \begin{tikzpicture}[scale=0.6,-,>=stealth',shorten >=.2pt,auto,
     semithick, initial text=, ampersand replacement=\&,]

  \matrix[nodes={draw, fill=none, shape=rectangle, minimum height=.2cm, minimum width=.2cm, align=center}, row sep=.6cm, column sep=.6cm] {
    \node[draw=none] (r) {$r_k$};
    
   \& \node[circle,add={-}{+}{}{}] (circle) {};
   \node[draw=none] (ez) at ([xshift=1cm,yshift=.15cm]circle)  {$u_k$};
   \node[rectangle,draw,minimum width=1cm,minimum height=1cm] (Kd) at ([xshift=0,yshift=-1.5cm]circle)  {\sc $\mat{K}_d$};
   \coordinate (kdsouth) at ([yshift=-2cm]Kd);

   \& complexnode/.pic={ 
      \node[rectangle,dashed,draw,minimum width=3cm,minimum height=1.6cm,label=\textbf{DAC}] (dac) {};
     \node[circle,add={+}{+}{}{},fill=gray!20] (q2) at ([xshift=-.65cm]dac.center) {};
     \node[draw=none] (q2t)  at ([yshift=.55cm]q2) {{\sc Q2}};
     \node[draw=none] (v2)  at ([yshift=-1.5cm]q2) {$\nu_2$};
     \node[fill=gray!20] (zoh) at ([xshift=.65cm]dac.center) {\sc ZOH};
   }   

   \& complexnode/.pic={ 
      \node[rectangle,dashed,draw,minimum width=8cm,minimum height=3.5cm,label=\textbf{Plant}] (plant)  at ([yshift=-.5cm]dac.center) {};
      \node[rectangle,draw,minimum width=1cm,minimum height=1cm] (B) at ([xshift=-2.5cm,yshift=.5cm]plant.center)  {\sc $\mat{B}$};
      \node[draw=none] (u) at ([xshift=-1cm,yshift=.15cm]B)  {$u(t)$};
      \node[circle,add={+}{+}{}{}] (p1) at ([xshift=-1.3cm,yshift=.5cm]plant.center) {};
      \node[draw=none] (xdot) at ([xshift=.85cm,yshift=.15cm]p1)  {$\dot{\vec{x}}(t)$};   
      \node[rectangle,draw,minimum width=1cm] (int) at ([xshift=.5cm,yshift=.5cm]plant.center) {\sc $\int$};
      \coordinate (xsouth) at ([xshift=1cm]int);
      \node[draw=none] (x) at ([xshift=1cm,yshift=.15cm]int)  {$\vec{x}(t)$};
      \node[rectangle,draw,minimum width=1cm,minimum height=1cm] (A) at ([xshift=.5cm,yshift=-1cm]plant.center)  {\sc $\mat{A}$};
      \coordinate (aeast) at ([xshift=1cm]A);
      \coordinate (awest) at ([xshift=-1.8cm]A);
    }   
        
   \& complexnode/.pic={ 
     \node[rectangle,dashed,draw,minimum width=3.5cm,minimum height=1.6cm,label=\textbf{ADC},] (adc) {};
     \draw[] ([xshift=-1cm]adc.center) -- ++(0.5,0.2cm);
     \coordinate (switch1) at ([xshift=-1cm]adc.center);
     \coordinate (switch2) at ([xshift=-0.4cm]adc.center);
     \node[circle,add={+}{+}{}{},fill=gray!20] (q1) at ([xshift=.6cm]adc.center) {};
     \node[draw=none] (q2t)  at ([yshift=.55cm]q1) {\sc Q1};
     \node[draw=none] (v1)  at ([yshift=-1.5cm]q1) {$\nu_1$};
     \node[draw=none] (y) at ([xshift=.85cm,yshift=.15cm]q1)  {$\vec{x}_k$};       
     \coordinate (ykeast) at ([xshift=1.9cm]q1);
     \coordinate (yksouth) at ([xshift=1.9cm,yshift=-3.5cm]q1);
   } 
   \& \coordinate (aux1);
   \& \\
  };

  \path[->] (v1) edge (q1.south);
  \path[->] (v2) edge (q2.south);
  \path[->] (r) edge (circle.west);
  \path[->] (circle.east) edge (q2.west);
  \path       (q2.east) edge (zoh.west);
  \path[->] (zoh.east) edge (B.west);
  \path
   (B.east) edge (p1.west)
   (p1.east) edge (int.west)
   (xsouth) edge (aeast)
   (aeast) edge (A.east)
   (A.west) edge (awest)
   (awest) edge (p1.south)
   (int.east) edge (switch1.west)
   (switch2) edge (q1.west);
  \path
   (q1.east) edge (ykeast)
   (ykeast) edge (yksouth)
   (yksouth) edge (kdsouth);
  \path[->]  (kdsouth) edge (Kd.south);
  \path (Kd.north) edge (circle.south);
 \end{tikzpicture}
}
 \caption{Closed-loop digital control system 
 \label{fig:digitalsystem}}
\end{figure*}
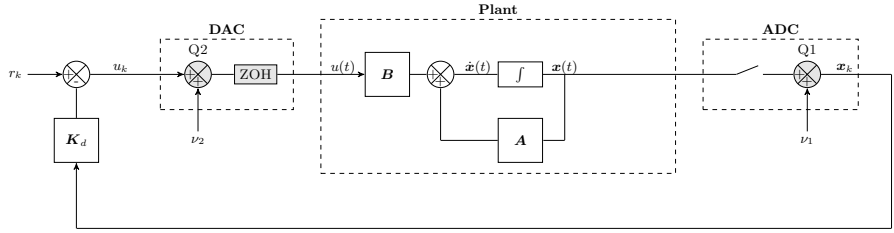

\subsection{Controller synthesis via state feedback}
\label{ssec:statefeedbackcontrol}

Models \eqref{eq:ode} and \eqref{eq:plant} depend on external non-determinism in the form of input signals $u (t)$ and  $u_k$, respectively. 
Feedback architectures can be employed to manipulate the properties and behaviors of the continuous process (the plant).   
We are interested in the synthesis of digital feedback control algorithms, 
as implemented on Field-Programmable Gate Arrays or Digital Signal Processors. 
The most basic feedback architecture is the state feedback one, 
where the control action $u_k$ (notice we work with the discretized signal) is computed by: 
\begin{equation}
\label{eq:controlaction}
u_k = r_{k} - K x_k. 
\end{equation}
Here, $K \in \mathbb{R}^{m \times n}$ is a state-feedback gain matrix, 
and $r_{k}$ is a reference signal (again digital).   
The closed-loop model then takes the form 
\begin{align}
\label{eq:closedloopss}
x_{k+1} = ( A_d - B_d K ) x_k + B_d r_k.
\end{align}
The gain matrix $K$ can be set so that the closed-loop discrete dynamics are
shaped as desired, for instance according to a specific stability goal or
around a specific dynamical behavior \cite{astrom1997computer}.  As argued
later in this work, we will target more complex objectives, such as
quantitative safety requirements, which are not typical in the digital
control literature.  Further, we will embrace the digital nature of the
controller, which manipulates quantized signals as discrete quantities represented with finite precision. 

\subsection{Stability of closed-loop systems}
\label{ssec:stability}

In this work we employ asymptotic stability in the CEGIS loop,  
as an objective for guessing controllers that are later proven sound over safety requirements.  
Asymptotic stability is a property that amounts to convergence of the model executions to an equilibrium point, 
starting from any states in a neighborhood of the point (see Figure~\ref{fig:ct} for the portrait of a stable execution, converging to the origin).  
In the case of linear systems as in~\eqref{eq:closedloopss}, 
considered with a zero reference signal, 
the equilibrium point of interest is the origin. 

A discrete-time LTI system as \eqref{eq:closedloopss} is asymptotically
stable if all the roots of its characteristic polynomial (i.e., the
eigenvalues of the closed-loop matrix $A_d - B_d K$) are inside the unity
circle of the complex plane, i.e., their absolute values are strictly less than
one~\cite{astrom1997computer} (this simple sufficient condition can be generalised, 
however this is not necessary in our work).  
In this paper, we express this stability specification $\phi_\mathit{stability}$ in terms of a check known as
\emph{Jury's criterion}~\cite{fadali}: this is an easy algebraic formula to
select the entries of matrix $K$ so that the closed-loop dynamics are shaped
as desired.

\subsection{Safety specifications for dynamical systems}
\label{ssec:safety}

We are not limited to the synthesis of digital stabilizing controllers -- a
well known task in the literature on digital control systems -- but target
safety requirements with an overall approach that is sound and automated. 
More specifically, we require that the closed-loop system
\eqref{eq:closedloopss} meets given safety specifications.  A safety
specification gives raise to a requirement on the states of the model, such
that the feedback controller (namely the choice of the gains matrix~$K$)
must ensure that the state never violates the requirement.  Note that a
stable, closed-loop system is not necessarily a safe system: indeed, the
state values may leave the safe part of the state space while they converge
to the equilibrium, which is typical in the case of oscillatory dynamics. 
In~this work, the safety property is expressed as:
\begin{equation}
\label{eq:safetyliteral}
\phi_\mathit{safety}\iff \forall k\ge 0.\, \bigwedge_{i=1}^{n}{\underline{x_{i}} \leq x_{i,k} \leq \overline{x_{i}}},
\end{equation}
%
%
where $\underline{x_{i}}$ and $\overline{x_{i}}$ are lower and upper bounds
for the $i$-th coordinate $x_{i}$ of state $x\in \mathbb R^n$ at the $k$-th
instant, respectively.  This means that the states will always be within an $n$-dimensional hyper-box.

Furthermore, it is practically relevant to consider the 
constraints $\phi_\mathit{input}$ on the input
signal $u_{k}$ and $\phi_\mathit{init}$ on the initial states $x_0$,
which we assume have given bounds:
$\phi_\mathit{input} = {\forall k.\underline{u} \leq u_{k} \leq \overline{u}} $, 
$\phi_\mathit{init} = \bigwedge_{i=1}^{n} \underline{x_{i,0}} \leq x_{i,0} \leq \overline{x_{i,0}}.$
For the former, this means that the control input might saturates in view of
physical constraints.


\subsection{Numerical representation and soundness} 
\label{sec:numeric_rep}

The models we consider have two sources of error that are due to numerical  
representation.  The first is the numerical error introduced by the
fixed-point numbers employed to model the plant, i.e., to represent the
plant dynamics $A_d$, $B_d$ and $x_k$.  The second is the quantization error
introduced by the digital controller, which performs operations on
fixed-point numbers.  In this section we outline the notation for the
fixed-point representation of numbers, and briefly describe the errors
introduced.  A~formal discussion is in
Appendix~\ref{appendix:numerical_errors}.

Let $\mathcal{F}_{\langle I,F \rangle}(x)$ denote a real number $x$
represented in a fixed point domain, with $I$ bits representing the integer
part and $F$ bits representing the decimal part.  The smallest number that
can be represented in this domain is $c_m=2^{-F}$.  Any mathematical
operations performed at the precision $\mathcal{F}_{\langle I,F \rangle}(x)$
will introduce errors, for which an upper bound can be
given~\cite{DBLP:conf/arith/BrainTRW15}.

We will use $\mathcal{F}_{\langle I_c,F_c \rangle}(x)$ to denote a real
number $x$ represented at the fixed-point precision of the controller, and
$\mathcal{F}_{\langle I_p,F_p \rangle}(x)$ to denote a real number $x$
represented at the fixed-point precision of the plant model ($I_c$ and $F_c$ are determined by the controller. We pick $I_p$ and $F_p$ for our synthesis such that $I_p \geq I_c$ and $\allowbreak F_p \geq F_c$).  Thus any mathematical operations in our modelled
digital controller will be in the range of $\mathcal{F}_{\langle I_c,F_c
\rangle}$, and all other calculations in our model will be carried out in the range of
$\mathcal{F}_{\langle I_p,F_p \rangle}$.  
The physical plant operates in the
reals, which means our verification phase must also account for the numerical error and quantization errors caused by representing the physical plant at the finite precision $\mathcal{F}_{\langle I_p,F_p \rangle}$.

\subsubsection{Effect on safety specification and stability}

Let us first consider the effect of the quantization errors on safety. 
Within the controller, state values are manipulated at low precision,
alongside the vector multiplication $Kx$.
%
The inputs are computed using the following equation: 
\begin{align*}
u_{k}&=-(\mathcal{F}_{\langle I_c,F_c \rangle}(K)\cdot\mathcal{F}_{\langle I_c,F_c \rangle}(x_{k})). 
\end{align*}

This induces two types of the errors detailed above: first, the truncation
error due to representing $x_k$ as $\mathcal{F}{\langle I_c,F_c
\rangle}(x_{k})$; and second, the rounding error introduced by the
multiplication operation.  We represent these errors as non-deterministic
additive noise.

An additional error is due to the representation of the plant dynamics, namely 
\begin{align*}
x_{k+1} &=\mathcal{F}_{\langle I_p,F_p \rangle}(A_d) \mathcal{F}_{\langle I_p,F_p \rangle}(x_{k}) + \mathcal{F}_{\langle I_p,F_p \rangle}(B_d)\mathcal{F}_{\langle I_p,F_p \rangle}(u_{k}).
\end{align*}
We address this error by use of interval
arithmetic~\cite{moore1966interval} in the verification phase.

Previous studies~\cite{gangli1} show that the FWL affects the poles and
zeros positions, degrading the closed-loop dynamics, causing steady-state
errors (see Appendix~\ref{sec:appendix:LTIbackground} for details) and
eventually de-stabilizing the system~\cite{Bessa16}.  However, since in this
paper we require stability only as a precursor to safety, it is sufficient
to check that the (perturbed, noisy) model converges to a neighborhood of
the equilibrium within the safe set (see Appendix~\ref{sec:stab_FWL}).

In the following, we shall disregard these steady-state errors (caused by
FWL effects) when stability is ensured by synthesis, and then verify its
safety accounting for the finite-precision errors.

\section{CEGIS of Safe Controllers for LTI Systems} 
\label{sec:CEGARIS} 

In this section, we describe our technique for synthesizing safe digital
feedback controllers using CEGIS.  For this purpose, we first provide the
synthesizer's general architecture, followed by describing our two
approaches to synthesizing safe controllers: the first one is a baseline
approach that relies on a na\"ive unfolding of the transition relation,
whereas the second uses abstraction to evaluate all possible executions of
the system.


 
\subsection{General architecture of the program synthesizer}
\label{synthesizer-general}

The input specification provided to the program synthesizer is of the form
$\exists P .\, \forall a.\, \sigma(a, P)$, where $P$ ranges over functions
(where a function is represented by the program computing it),
$a$ ranges over ground terms, and $\sigma$ is a quantifier-free formula. 
We~interpret the ground terms over some finite domain~$\mathcal{D}$.
The design of our synthesizer consists of two phases, an inductive synthesis
phase and a validation phase, which interact via a finite set of test
vectors {\sc inputs} that is updated incrementally.  Given the
aforementioned specification $\sigma$, the inductive synthesis procedure
tries to find an existential witness $P$ satisfying the specification
$\sigma(a, P)$ for all $a$ in {\sc inputs} (as opposed to all $a \in
\mathcal{D}$).
If the synthesis phase succeeds in finding a witness~$P$, this witness is a
candidate solution to the full synthesis formula.  We pass this candidate
solution to the validation phase, which checks whether it is a full solution
(i.e., $P$ satisfies the specification $\sigma(a, P)$ for all
$a\in\mathcal{D}$).  If this is the case, then the algorithm terminates. 
Otherwise, additional information is provided to the inductive synthesis
phase in the form of a new counterexample that is added to the {\sc inputs}
set and the loop iterates again.  More details about the general
architecture of the synthesizer can be found
in~\cite{DBLP:conf/lpar/DavidKL15}.
%

\subsection{Synthesis problem: statement (recap) and connection to program synthesis}


At this point, we recall the the synthesis problem that we solve in this
work: we seek a digital feedback controller $K$ (see
Eq.~\ref{eq:controlaction}) that makes the closed-loop plant model safe for
initial state $x_0$, reference signal $r_k$ and input $u_k$ as defined in
Sec.~\ref{ssec:safety}.  We consider non-deterministic initial states within
a specified range, the reference signal to be set to zero, saturation on the
inputs, and account for digitization and quantization errors introduced by
the controller.

When mapping back to the notation used for describing the general
architecture of the program synthesizer, the controller $K$ denotes $P$,
$(x_0, u_k)$ represents $a$ and
$\phi_\mathit{stability} \wedge \phi_\mathit{input} \wedge
\phi_\mathit{init} \wedge \phi_\mathit{safety}$
denotes the specification $\sigma$.

\subsection{Na\"ive Approach: CEGIS with multi-staged verification}
\label{sec:CEGIS-precision-incrementation}

\begin{figure}[htb]
{\scriptsize
\centering
{
 \begin{tikzpicture}[scale=0.3,->,>=stealth',shorten >=.2pt,auto, semithick, initial text=, ampersand replacement=\&,]
  \matrix[nodes={draw, fill=none, shape=rectangle, minimum height=.2cm, minimum width=.2cm, align=center
},
          row sep=.6cm, column sep=.9cm] {
   \coordinate (aux1);
   \& \coordinate (aux2);
   \&;\\
   \coordinate (aux3);
   \& \coordinate (aux4);
   \&;\\
   \coordinate (aux5);
   \& \coordinate (aux6);
   \&;\\
   \node[minimum width=1.5cm, minimum height=0.6cm, fill=gray!20] (synth) {{\sc 1.Synthesize}};
   \&
   complexnode/.pic={ 
     \node[rectangle,draw,dotted,
	minimum width=6cm,
	minimum height=1cm,
        pattern=north west lines, pattern color=gray!20,
	label={\sc ~~~~~~~~~~~~Verify},] (verif) {};
     \node[minimum width=1cm, minimum height=0.6cm, fill=gray!20] (verif1) at ([xshift=-2cm]verif.center) {{\sc 2.Safety}};
     \node[minimum width=1cm, minimum height=0.6cm, fill=gray!20] (verif2) at ([xshift=0cm]verif.center) {{\sc 3.Precision}};
     \node[minimum width=1cm, minimum height=0.6cm, fill=gray!20] (verif3) at ([xshift=2cm]verif.center) {{\sc 4.Complete}};
   } 
   \& \node[ellipse, fill=gray!20] (done) {{\sc Done}};\\
   \& \\
   \node[minimum height=0cm] (gp) {\sf Program Search};
   \&
   complexnode/.pic={ 
     \coordinate (aux);
   \node (bmc) at ([xshift=-2cm]aux.center) {\sf BMC-based \\ \sf Verifier};
   \node (fp)  at ([xshift=0cm]aux.center) {\sf Fixed-point \\ \sf Arithmetic\\ \sf Verifier};
   \node (sv)  at ([xshift=2cm]aux.center) {\sf Completeness\\ \sf Verifier};
   }   
    \\
  };

   \path
    ([yshift=2em]synth.east) edge node[xshift=-0.5em,align=center] {$K$} ([yshift=2em]verif1.west)
    ([yshift=-2em]verif1.west) edge node {C-ex} ([yshift=-2em]synth.east)
    ([xshift=-5em]fp.north) edge node[align=center]  {T/F} ([xshift=-5em]verif2.south)
    ([xshift=-5em]sv.north) edge node[align=center]  {T/F} ([xshift=-5em]verif3.south)
    ([xshift=5em]verif1.south) edge node[align=center] {$K$} ([xshift=5em]bmc.north)
    ([xshift=5em]verif2.south) edge node[align=center] {$K$} ([xshift=5em]fp.north)
    ([xshift=5em]verif3.south) edge node[align=center] {$K$} ([xshift=5em]sv.north)
    ([xshift=-5em]bmc.north) edge node[align=center]  {UNSAT/\\model} ([xshift=-5em]verif1.south)
    (verif) edge node {PASS} (done)
    ([xshift=5em]synth.south) edge node[align=center] {Inputs} ([xshift=5em]gp.north)
    ([xshift=-5em]gp.north) edge node[align=center] {UNSAT/\\$K$} ([xshift=-5em]synth.south)
    (aux3) edge (synth.north);
   \path[-]
   (verif2.north) edge node[align=center] {} ([xshift=0cm]aux6)
   ([xshift=0cm]aux6) edge node[align=center] {Increase Precision} (aux5)
   (verif3.north) edge node[align=center] {} ([xshift=6.7cm]aux4)
   ([xshift=6.7cm]aux4) edge node[align=center] {Increase Unfolding Bound} (aux3);

 \end{tikzpicture}
}}
\caption{CEGIS with multi-staged verification}
\label{fig:CEGIS-precision-increment}
\end{figure}
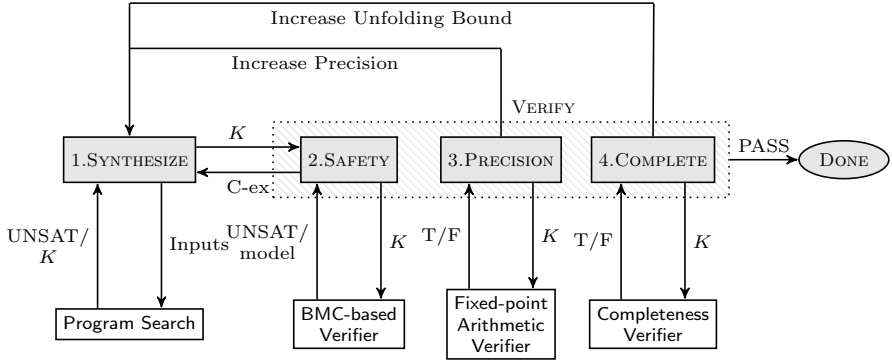


An overview of the algorithm for controller synthesis is given in
Fig.~\ref{fig:CEGIS-precision-increment}.  One important observation is that
we verify and synthesize a controller over $k$ time steps.  We~then compute
a completeness threshold $\overline{k}$~\cite{DBLP:conf/vmcai/KroeningS03}
for this controller, and verify correctness for $\overline{k}$ time steps. 
Essentially, $\overline{k}$ is the number of iterations required to
sufficiently unwind the closed-loop state-space system, 
which ensures that the boundaries are not violated for any other $k{>}\overline{k}$.

\begin{theorem} There exists a finite $\overline{k}$ such that it is
sufficient to unwind the closed-loop state-space system up to $\overline{k}$
in order to ensure that $\phi_\mathit{safety}$ holds. 
\end{theorem}

\begin{proof}
A stable control system is known to have converging dynamics.  Assume the
closed-loop matrix eigenvalues are not repeated (which is sensible to do,
since we select them).  The distance of the trajectory from the reference
point (origin) decreases over time within subspaces related to real-valued
eigenvalues; however, this is not the case in general when dealing with
complex eigenvalues.  Consider the closed-loop matrix that updates the
states in every discrete time step, and select the eigenvalue $\vartheta$
with the smallest (non-trivial) imaginary value.  Between every pair of
consecutive time steps $k\,T_s$ and $(k+1)\,T_s$, the dynamics projected on
the corresponding eigenspace rotate $\vartheta T_s$ radians.  Thus, taking
$\overline{k}$ as the ceiling of $\frac{2\pi}{\vartheta T_s}$, after
$k{\geq}\overline{k}$ steps we have completed a full rotation, which results
in a point closer to the origin.  The synthesized $\overline{k}$ is the completeness threshold.
\qed 
\end{proof}

\medskip

Next, we describe the different phases in Fig.~\ref{fig:CEGIS-precision-increment}
(blocks 1 to 4) in detail.

\begin{enumerate}

\item The inductive synthesis phase ({\sc synthesize}) uses BMC to
compute a candidate solution $K$ that satisfies both the stability criteria
(Sec.~\ref{ssec:stability}) and the safety specification
(Sec.~\ref{ssec:safety}).  To synthesize a controller that satisfies the
stability criteria, we require that a computed polynomial
satisfies Jury's criterion~\cite{fadali}.  The details of this calculation can be found in
the Appendix.

Regarding the second requirement, we synthesize a safe controller by
unfolding the transition system $k$ steps and by picking a controller $K$
and a single initial state, such that the states at each step do not violate
the safety criteria.  That is, we ask the bounded model checker if there
exists a $K$ that is safe for at least one $x_0$ in our set of all possible
initial states.  This is sound if the current $k$ is greater than the completeness
threshold.  We~also assume some precision $\langle I_p,F_p\rangle$ for the
plant and a sampling rate.  The checks that these assumptions hold are
performed by subsequent {\sc verify} stages.


\begin{algorithm}[]
\begin{algorithmic}[1]
\Function{$\mathit{safetyCheck}()$}{}
\State assert($ \underline{u}  \leq u \leq \overline{u}$)
 \State set $x_0$ to be a vertex state, e.g., $[\underline{x_0},\underline{x_0}]$	
\For {($c=0;~c < 2^\mathit{Num\_States};~c\mbox{++}$)}
	\For{($i=0;~i< k;~i\mbox{++}$)}
		\State $u = (plant\_typet)((controller\_typet)K * (controller\_typet) x)$
		\State $x = A * x + B * u$
		\State assert($\underline{x} \leq x \leq \overline{x}$ )
  	\EndFor
  	\State set $x_0$ to be a new vertex state
  	\EndFor
\EndFunction
\end{algorithmic}
\caption{Safety check\label{alg:safetycheck}}
\end{algorithm}

\item The first {\sc verify} stage, {\sc safety}, checks that the candidate
solution $K$, which we synthesized to be safe for at least one initial
state, is safe for \emph{all} possible initial states, i.e., does not reach
an unsafe state within $k$ steps where we assume $k$ to be under the
completeness threshold.  After unfolding the transition system corresponding
to the previously synthesized controller $k$ steps, we check that the safety
specification holds for any initial state. This is shown in Alg.~~\ref{alg:safetycheck}.

\item The second {\sc verify} stage, {\sc precision}, 
 restores soundness with respect to the plant's precision
by using interval arithmetic \cite{moore1966interval} to validate the 
operations performed by the previous stage. 

\item The third {\sc verify} stage, {\sc complete}, checks that the current
$k$ is large enough to ensure safety for any $k'{>}k$.  Here, we compute the
completeness threshold $\overline{k}$ for the current candidate controller $K$ and
check that $k{\geq}\overline{k}$.  This is done according to the argument
given above and illustrated in Fig.~\ref{fig:ct}.

\end{enumerate}

\begin{figure*}[t]
\centering
\includegraphics[width=0.6\textwidth]{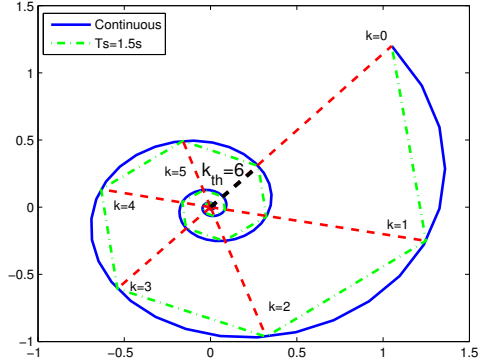}
\vspace{0.1cm}
\caption{Completeness threshold for multi-staged verification. $T_s$ is the time step for the time discretization of the control matrices.}
\label{fig:ct}
\end{figure*}

Checking that the safety specification holds for any initial state can be 
computationally expensive if the bounds on the
allowed initial states are large. 

\begin{theorem} 
If a controller is safe for each of the corner cases of our hypercube of
allowed initial states, it is safe for any initial state in the hypercube. 
\end{theorem}

Thus we only need to check $2^n$ initial states, where $n$
is the dimension of the state space (number of continuous variables). 
\begin{proof}
Consider the set of initial states, $X_0$, which we assume to be convex since it is a hypercube. 
Name $v_i$ its vertexes, where $i=1,\ldots, 2^n$.  
Thus any point $x \in X_0$ can be expressed by convexity as $x = \sum_{i=1}^{2^n} \alpha_i v_i$, 
where $\sum_{i=1}^{2^n} \alpha_i =1$. Then if $x_0=x$, we obtain 
\begin{align*}
x_k   &= (A_d - B_d K)^k x = (A_d - B_d K)^k \sum_{i=1}^{2^n} \alpha_i v_i 
      = \sum_{i=1}^{2^n} \alpha_i (A_d - B_d K)^k v_i = \sum_{i=1}^{2^n} \alpha_i x_k^i, 
 \end{align*}     
where $x_k^i$ denotes the trajectories obtained from the single vertex
$v_i$.  We~conclude that any $k$-step trajectory is encompassed, within a
convex set, by those generated from the vertices. 
\qed
\end{proof}

\subsubsection{Illustrative Example} \label{sec:running-ex}

We illustrate our approach with an example,
extracted from~\cite{Franklin15}.
Since we have not learned any information about the system yet,
we pick an arbitrary candidate solution (we always choose $K=[0
  \ 0 \ 0]^T$ in our experiments to simplify reproduction), and a
precision of $I_p=13$, $F_p=3$.  In the first
{\sc verify} stage, the {\sc safety} check finds the counterexample
$ x_0 = [-0.5 \ 0.5 \ 0.5] $.
After adding the new counterexample to its sets of {\sc inputs}, {\sc
  synthesize} finds the candidate solution $K=[0\,0
  \,0.00048828125]^T$, which prompts the {\sc safety} verifier to
return $x_0=[-0.5\,-0.5\,-0.5]$ as the new counterexample.

In the subsequent iteration, the synthesizer is unable to find further 
suitable candidates and it returns UNSAT, meaning that the current precision is
insufficient.  Consequently, we increase the precision the plant is modelled with to $I_p=17$,
$F_p=7$. We increase the precision by 8 bits each step in order to be compliant with the CBMC type API.
Since the previous counterexamples were obtained at lower precision,
we remove them from the set of counterexamples.  Back in the {\sc
  synthesize} phase, we re-start the process with a candidate solution
with all coefficients $0$.  Next, the {\sc safety} verification stage
provides the first counterexample at higher precision, $x_0=[-0.5
  \ 0.5 \ 0.5]$ and {\sc synthesize} finds $K=[0 \ 0.01171875
  \ 0.015625]^T$ as a candidate that eliminates this counterexample.
However, this candidate triggers the counterexample
$x_0=[0.5\ -0.5\ -0.5]$ found again by the {\sc safety} verification
stage.  In the next iteration, we get the candidate $K=[0 \ 0
  \ -0.015625]$, followed by the counterexample $x_0 = [0.5 \ 0.5
  \ 0.5]$. Finally, {\sc synthesize} finds the candidate $K=[0.01171875
  \ -0.013671875 \ -0.013671875]^T$, which is validated as a final
solution by all verification stages.

\subsection{Abstraction-based CEGIS}
\label{sec:CEGIS-abstraction-refinement}



  The na\"ive approach described in Sec.~\ref{sec:CEGIS-precision-incrementation}
  synthesizes a controller for an individual initial state
and input with a bounded time horizon and, subsequently, it generalizes it to all reachable states,
inputs, and time horizons during the verification phase.
Essentially, this approach relies on the symbolic
simulation over a bounded time horizon of individual initial states and inputs that form
part of an uncountable space and tries to generalize it for an
infinite space over an infinite time horizon.




Conversely, in this section, 
we  find a controller for a continuous initial set of states and set of inputs, 
  over an abstraction of the continuous dynamics \cite{cattaruzza2015unbounded} that conforms to
  witness proofs at specific times. 
  Moreover, this approach uses abstraction refinement enabling us to 
  always start with a very simple description regardless of the complexity of the overall
dynamics, and only expand to more complex models when a solution
cannot be found.

  The CEGIS loop for this approach is illustrated in Fig.~\ref{fig:CEGARIS}.


\begin{enumerate}
\item 
  We start by doing some preprocessing:
  \begin{enumerate}
\item Compute the characteristic polynomial of the matrix $(A_d-B_dK)$ as 
$P_a(z) = z^n+\sum_{i=1}^n{(a_i-k_i)z^{n-i}}$. 



\item Calculate the noise set $N$ from the quantizer resolutions and estimated round-off errors: 
$$N=\left \{ \nu_1+\nu_2+ \nu_3 : \nu_1 \in \left[-\frac{q_1}{2}\ \ \frac{q_1}{2}\right] 
\wedge \nu_2 \in \left[-\frac{q_2}{2}\ \ \frac{q_2}{2}\right]  \wedge \nu_3 \in \left[-q_3\ \ q_3\right]  \right \}\nonumber$$
where  $q_1$ is the error introduced by the truncation in the ADC, $q_2$ is
the error introduced by the DAC and $q_3$ is the maximum truncation and
rounding error in $u_k=-K \cdot \mathcal{F}_{\langle I_c,F_c \rangle}(x_k)$ as
discussed in Section~\ref{sec:numeric_rep}.  More details on how to model
quantization as noise are given in
Appendix~\ref{appendix:quantization-noise}.

\item Calculate a set of initial bounds on $K$, $\phi_\mathit{init}^{K}$,
based on the input constraints 
$$(\phi_\mathit{init} \wedge \phi_\mathit{input} \wedge u_k=-K x_k)
\Rightarrow \phi_\mathit{init}^{K}$$
Note that these bounds will be used by the {\sc synthesize} phase to reduce the size of the solution space. 

\end{enumerate}
\item In the {\sc synthesize} phase, we synthesize a candidate controller
  $K \in \mathbb{R}\langle I_c,F_c\rangle^n$ that satisfies
  $\phi_\mathit{stability} \wedge \phi_\mathit{safety} \wedge \phi_\mathit{init}^{K}$ by invoking a SAT solver.
If there is no candidate solution we return UNSAT and exit the loop.
\item Once we have a candidate solution, we perform a safety verification 
  of the 
  progression of the system from $\phi_\mathit{init}$ over time,
$x_{k} \models \phi_\mathit{safety}$. 
  In order to compute the progression of point $x_0$ at iteration $k$,
  we accelerate the dynamics of the closed-loop system and obtain:
%
\begin{align}
\hspace{-0.1in} x=&(A_d-B_dK)^kx_0
+\sum_{i=0}^{k-1} (A_d-B_dK)^i B_{n}(\nu_1+\nu_2+\nu_3) : B_n= [1 \cdots 1]^T
\end{align}
As this still requires us to verify the system for every $k$ up to infinity,
we use abstract acceleration again to obtain the reach-tube, i.e., the set
of all reachable states at all times given an initial set
$\phi_\mathit{init}$:
\begin{align}
\label{eq:aa_observer_LTI_cf}
\hat{X}^\#
=\mathcal{A} X_0 + \mathcal{B}_{n} N, \quad
X_0 =\left \{x : x \models \phi_\mathit{init} \right\}, 
\end{align}
where $\mathcal{A}=\bigcup_{k=1}^\infty (A_d-B_dK)^k,
\mathcal{B}_{n}=\bigcup_{k=1}^\infty \sum_{i=0}^k(A_d-B_dK)^iB_{n}$ are
abstract matrices for the closed-loop system~\cite{cattaruzza2015unbounded},
whereas the set $N$ is non-deterministically chosen.

We next evaluate $\hat{X}^\# \models \phi_\mathit{safety}$.  If the
verification holds we have a solution, and exit the loop.  Otherwise, we
find a counterexample iteration $k$ and corresponding initial point $x_0$
for which the property does not hold, which we use to locally refine the
abstraction.  When the abstraction cannot be further refined, we provide
them to the {\sc abstract} phase.
%
\item If we reach the {\sc abstract} phase, it means that the candidate solution is not valid,
  in which case we must refine the abstraction used by the synthesizer.
\begin{enumerate}
\item Find the constraints that invalidate the property
 as a set of counterexamples for the eigenvalues, which we define as $\phi_\Lambda$. This is a constraint in the spectrum i.e., transfer function) of the closed loop dynamics. 
\item We use $\phi_\Lambda$ to
  further constrain the characteristic polynomial 
$z^n+\sum_{i=1}^n(a_i-k_i)z^{n-i}=\prod_{i=1}^n (z-\lambda_i) : |\lambda_i|<1 \wedge \lambda_i \models \phi_{\Lambda}$. These constraints correspond to specific iterations for which the system may be unsafe.
\item Pass the refined abstraction $\phi(K)$ with the new constraints and the list of iterations $k$ to the {\sc synthesize} phase.
\end{enumerate} 
\end{enumerate}

\subsubsection{Illustrative Example}
Let us consider the following example with discretized dynamics
$$A_d = \left[\begin{array}{ccc}2.6207&-1.1793&0.65705\\2&0&0\\0&0.5&0\end{array}\right],
\quad B_d = \left [\begin{array}{c}8\\0\\0\end{array}\right]$$
Using the initial state bounds $\underline{x_{0}}=-0.9$ and
$\overline{x_{0}}=0.9$, the input bounds $\underline{u}=-10$ and
$\overline{u}=10$, and safety specifications $\underline{x_{i}}=-0.92$ and
$\overline{x_{i}}=0.92$, the {\sc synthesize} phase in
Fig.~\ref{fig:CEGARIS} generates an initial candidate controller
$K=[\begin{array}{ccc}0.24609375&-0.125&0.1484375\end{array}]$.
This candidate is chosen for its closed-loop stable dynamics, but the {\sc verify} phase finds it to be
unsafe and returns a list of iterations with an initial state that
fails the safety specification {\footnotesize $(k,x_0) {\in}
\{ (2, [\begin{array}{ccc}0.9&-0.9&0.9\end{array}]),\allowbreak (3,
[\begin{array}{ccc}0.9&-0.9&-0.9\end{array}])\}$}.  This allows the {\sc
abstract} phase to create a new safety specification that considers these
iterations for these initial states to constrain the solution space.  This
refinement allows {\sc synthesize} to find a new controller {\footnotesize
$K=[\begin{array}{ccc}0.23828125&-0.17578125&0.109375\end{array}]$}, which
this time passes the verification phase, resulting in a safe system.


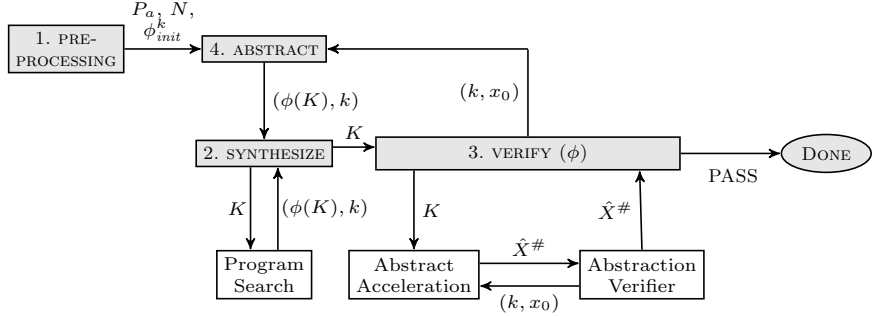
\begin{figure}
\centering
{\scriptsize
\begin{tikzpicture}[scale=0.3,->,>=stealth',shorten >=.2pt,auto, semithick, ampersand replacement=\&,]
  \matrix[nodes={draw, fill=none, shape=rectangle, minimum height=.2cm, minimum width=.2cm, align=center},row sep=.8cm, column sep=.2cm] {
   \coordinate (aux1);
   \& \coordinate (aux2);
   \& ;\\
   \&\node[fill=gray!20,align=center,xshift=-1.5cm] (pre) {{\sc 1. pre-}\\{\sc processing}};
   \& \node[fill=gray!20,align=center] (abstract) {\sc 4. abstract};
   \& \coordinate (aux); \\ 
   \&
   \& \node[fill=gray!20,align=center] (synth) {\sc 2. synthesize};
   \& \node[fill=gray!20,align=center, minimum width=4cm] (verify) {\sc 3. verify ($\phi$)};
      \node[draw=none] (SAT) at ([xshift=2.7cm,yshift=-.3cm]verify)  {\sc PASS};
   \& \node[ellipse, fill=gray!20] (done) {{\sc Done}};\\
   \&
   \& \node[draw,rectangle,align=center] (KSAT) {Program \\ Search};
   \& complexnode/.pic={
     \coordinate (AA);
     \node[draw,rectangle,align=center] (AAV) at ([xshift=-1.5cm]AA.center) {Abstract\\Acceleration};
     \node[draw,rectangle,align=center] (AAC) at ([xshift=1.5cm]AA.center) {Abstraction \\ Verifier};
    }\\
  };
  \path
    (pre.east) edge node[align=center] {$P_a$, $N$,\\ $\phi_\mathit{init}^k$} (abstract.west)
    (abstract.south) edge node{$(\phi(K), k)$} (synth.north)
    ([yshift=1em]synth.east) edge node {$K$} ([yshift=1em]verify.west)
    (aux) edge (abstract.east) 
    (verify.east) edge (done.west);
  \path
    ([xshift=+.6cm]KSAT.north) edge[left] node[xshift=-.3cm] {$K$} ([xshift=0.6cm]synth.south) 
    ([xshift=-.6cm]synth.south) edge node[align=center,xshift=.3cm] {$(\phi(K),k)$} ([xshift=-.6cm]KSAT.north)
    ([yshift=.5cm]AAV.east) edge node{\sl $\hat{X}^\#$} ([yshift=.5cm]AAC.west)
    ([yshift=-.5cm]AAC.west) edge node{$(k,x_0)$} ([yshift=-.5cm]AAV.east)
    (AAC.north) edge node{\sl $\hat{X}^\#$} ([xshift=4.9cm]verify.south)
    ([xshift=-5cm]verify.south) edge node{$K$} (AAV.north);
  \path[-] 
     (verify.north) edge node {$(k,x_0)$} (aux);
\end{tikzpicture}
}
\caption{Abstraction-based CEGIS}
\label{fig:CEGARIS} 
\end{figure}

\section{Experimental Evaluation}

\subsection{Description of the benchmarks}
\label{experimental-setup}

A set of state-space models for different classes of systems has been taken 
from the literature~\cite{Franklin15, maglev, converters, CTMS} and employed
for validating our methodology.

\textit{DC Motor Rate} plants describes the angular velocity of a DC Motor, respectively. 
The \textit{Automotive Cruise System} plant represents the speed of a motor vehicle. 
The \textit{Helicopter Longitudinal Motion} plant provides the longitudinal motion model 
of a helicopter. 
The \textit{Inverted Pendulum} plant describes a pendulum model
with its center of mass above its pivot point. 
The \textit{Magnetic Suspension} plant provides a physical model for which 
a given object is suspended via a magnetic field. 
The \textit{Magnetized Pointer} plant describes a physical model employed in analogue gauges 
and indicators that is rotated through interaction with magnetic fields.
The \textit{1/4 Car Suspension} plant presents a physical model that connects a car to its wheels 
and allows relative motion between the two parts.
The \textit{Computer Tape Driver} plant describes a system to read and write data 
on a storage device.

Our benchmarks are SISO models (Section~\ref{sec:preliminaries}).  The
Inverted Pendulum appears to be a two-output system, but it is treated as
two SISO models during the experiments.  All the
state measurements are assumed to be available  
(current work targets the extension of our framework to observer-based synthesis).   

All benchmarks are discretized with different sample times~\cite{fadali}. 
All experiments are performed considering $\underline{x_{i}}=-1$ and 
$\overline{x_{i}}=1$ and the reference inputs $r_{k}=0, \forall k>0$.
We conduct the experimental evaluation on a 12-core 2.40\,GHz Intel
Xeon E5-2440 with 96\,GB of RAM and Linux OS.  We use the Linux
\emph{times} command to measure CPU time used for each benchmark.  The
runtime is limited to one hour per benchmark.

\subsection{Objectives}

Using the state-space models given in Section~\ref{experimental-setup}, 
our evaluation has the following two experimental goals:
%


\begin{enumerate}

\item[EG1] \textbf{(CEGIS)} Show that both the multi-staged and the abstraction-based
CEGIS approaches are able to generate FWL digital controllers in a reasonable amount of time.

\item[EG2] \textbf{(sanity check)} Confirm the
stability and safety of the synthesized controllers outside of our model.

\end{enumerate}

\subsection{Results}

We provide the results in Table~\ref{tab:results}.  Here \textit{Benchmark} is
the name of the respective benchmark, \textit{Order} is the number of
continuous variables, $\mathcal{F}_{\langle I_p,F_p \rangle}$ is the
fixed-point precision used to model the plant, while \textit{Time} is the
total time required to synthesize a controller for the given
plant with one of the two methods.  Timeouts are indicated by~\xmark.  The
precision for the controller, $\mathcal{F}_{\langle I_c,F_c \rangle}$, is
chosen to be $I_c = 8$, $F_c = 8$.

For the majority of our benchmarks, we observe that the abstraction-based
back-end is faster than the basic multi-staged verification approach, and
finds one solution more (9) than the multi-staged back-end (8).  In direct
comparison, the abstraction-based approach is on average able to find a
solution in approximately 70\% of the time required using the multi-staged
back-end, and has a median run-time 1.4\,s, which is seven times smaller
than the multi-staged approach.  The two back-ends complement each other in
benchmark coverage and together solve all benchmarks in the set.  On average
our engine spent 52\% in the synthesis and 48\% in the verification phase.

\begin{table}
\centering
\begin{tabular}{| r | l | c | c | r | c | r |}
\hline
\# & \multicolumn{1}{|c|}{Benchmark} & \multicolumn{1}{|c|}{Order} & \multicolumn{2}{|c|}{Multi-staged}                 & \multicolumn{2}{|c|}{Abstraction} \\
   &                                  & & \multicolumn{1}{|c|}{$\mathcal{F}_{\langle I_p,F_p \rangle}$} & \multicolumn{1}{|c|}{Time} & \multicolumn{1}{|c|}{$\mathcal{F}_{\langle I_p,F_p \rangle}$} & \multicolumn{1}{|c|}{Time} \\\hline
1  & Cruise Control  & 1 & 8,16   & 8.40\,s & 16,16  &   2.17\,s \\
2  & DC Motor          & 2 & 8,16   & 9.45\,s & 20,20  &   2.06\,s \\
3  & Helicopter        & 3 & \xmark & \xmark  & 16,16  &   1.37\,s \\
4  & Inverted Pendulum & 4 & 8,16   & 9.65\,s & 16,16  &   0.56\,s \\
5  & Magnetic Pointer  & 2 & \xmark & \xmark  & 28,28  &  44.14\,s \\
6  & Magnetic Suspension & 2 & 12,20  &10.41\,s & 16,16  &   0.61\,s \\
7  & Pendulum          & 2 & 8,16   &14.02\,s & 16,16  &   0.60\,s \\
8  & Suspension        & 2 & 12,20  &73.66\,s & \xmark &    \xmark \\
9  & Tape Driver       & 3 & 8,16   &10.10\,s & 16,16  &  68.24\,s \\
10 & Satellite          & 2 & 8,16   & 9.43\,s & 16,16  &   0.67\,s \\\hline
\end{tabular}
\vspace{0.05in}
\caption{Experimental results\label{tab:results}}
\end{table}

The median run-time for our benchmark set is 9.4\,s.  Overall, the average
synthesis time amounts to approximately 15.6\,s.  We consider these times
short enough to be of practical use to control engineers, and thus affirm
EG1.

There are a few instances for which the system fails to find a controller. 
For the na\"ive approach, the completeness threshold may be too large, thus
causing a timeout.  On the other hand, the abstraction-based approach may
require a very precise abstraction, resulting in too many refinements and,
consequently, in a timeout.  Yet another source of incompleteness is the
inability of the {\sc synthesize} phase to use a large enough precision for
the plant model.

The synthesized controllers are confirmed to be safe outside of our model
representation using MATLAB, achieving EG2.  A~link to the full
experimental environment, including scripts to reproduce the results, all
benchmarks and the tool, is provided in the
footnote as an Open Virtual Appliance
(OVA).\footnote{\url{www.cprover.org/DSSynth/controller-synthesis-cav-2017.tar.gz}}
The provided experimental environment runs multiple discretisations for each
benchmark, and lists the fastest as the result synthesis time.

\subsection{Threats to validity}

\textit{Benchmark selection:} We report an assessment of both our approaches
over a diverse set of real-world benchmarks.  Nevertheless, this set of
benchmarks is limited within the scope of this paper and the performance may
not generalize to other benchmarks.\\
\textit{Plant precision and discretization heuristics:} Our algorithm to
select suitable FWL word widths to model the plant behavior 
increases the precision by 8 bits at each step 
in order to be compliant with the CBMC type API. 
Similarly, for discretization, we run multiple discretizations for each
benchmark and retain the fastest run.
This works sufficiently well for our benchmarks, but
performance may suffer in some cases, for example if the completeness
threshold is high.\\
\textit{Abstraction on other properties:} The performance gain from abstract
acceleration may not hold for more complex properties than safety, for
instance ``eventually reach and always remain in a given safe set''.

\section{Conclusion}

We have presented two automated approaches to synthesize digital
state-feedback controllers that ensure both stability and safety over the
state-space representation.  The first approach relies on unfolding of the
closed-loop model dynamics up to a completeness threshold, while the second
one applies abstraction refinement and acceleration to increase speed whilst
retaining soundness.
Both approaches are novel within the control literature: they give a fully
automated synthesis method that is algorithmically and numerically sound,
considering various error sources in the implementation of the digital
control algorithm and in the computational modeling of plant dynamics.
Our experimental results show that both approaches are able to synthesize
safe controllers for most benchmarks within a reasonable amount of time
fully automatically.  In particular, both approaches complement each other
and together solve all benchmarks, which have been derived from the control
literature.

Future work will focus the extension of these approaches to the
continuous-time case, to models with output-based control architectures
(with the use of observers), and to the consideration of more complex specifications. 



\begin{thebibliography}{10}

\bibitem{CTMS}
Control tutorials for {MATLAB} and {SIMULINK}.
\newblock \url{http://ctms.engin.umich.edu/}.

\bibitem{hscc-paper}
A.~Abate, I.~Bessa, D.~Cattaruzza, L.~C. Cordeiro, C.~David, P.~Kesseli, and
  D.~Kroening.
\newblock Sound and automated synthesis of digital stabilizing controllers for
  continuous plants.
\newblock In {\em Hybrid Systems: Computation and Control (HSCC)}, pages
  197--206. ACM, 2017.

\bibitem{Anta2010}
A.~Anta, R.~Majumdar, I.~Saha, and P.~Tabuada.
\newblock Automatic verification of control system implementations.
\newblock In {\em EMSOFT}, pages 9--18, 2010.

\bibitem{astrom1997computer}
K.~{\AA}str{\"o}m and B.~Wittenmark.
\newblock {\em {Computer-controlled systems: theory and design}}.
\newblock Prentice Hall information and system sciences series. Prentice Hall,
  1997.

\bibitem{Bessa16}
I.~Bessa, H.~Ismail, R.~Palhares, L.~Cordeiro, and J.~E.~C. Filho.
\newblock Formal non-fragile stability verification of digital control systems
  with uncertainty.
\newblock {\em {IEEE} Transactions on Computers}, 66(3): pages 545--552, 2017.

\bibitem{DBLP:conf/arith/BrainTRW15}
M.~Brain, C.~Tinelli, P.~R{\"{u}}mmer, and T.~Wahl.
\newblock An automatable formal semantics for {IEEE-754} floating-point
  arithmetic.
\newblock In {\em {ARITH}}, pages 160--167. {IEEE}, 2015.

\bibitem{cattaruzza2015unbounded}
D.~Cattaruzza, A.~Abate, P.~Schrammel, and D.~Kroening.
\newblock Unbounded-time analysis of guarded {LTI} systems with inputs by
  abstract acceleration.
\newblock In {\em 22nd International Symposium on Static Analysis}, volume 9291
  of {\em LNCS}, pages 312--331, 2015.

\bibitem{DBLP:conf/lpar/DavidKL15}
C.~David, D.~Kroening, and M.~Lewis.
\newblock Using program synthesis for program analysis.
\newblock In {\em Logic for Programming, Artificial Intelligence, and Reasoning
  {(LPAR-20)}}, LNCS, pages 483--498. Springer, 2015.

\bibitem{daes20161}
I.~V. de~Bessa, H.~Ismail, L.~C. Cordeiro, and J.~E.~C. Filho.
\newblock Verification of fixed-point digital controllers using direct and
  delta forms realizations.
\newblock {\em Design Autom. for Emb. Sys.}, 20(2):95--126, 2016.

\bibitem{Duggirala2015}
P.~S. Duggirala and M.~Viswanathan.
\newblock Analyzing real time linear control systems using software
  verification.
\newblock In {\em IEEE Real-Time Systems Symposium}, pages 216--226, Dec 2015.

\bibitem{fadali}
S.~Fadali and A.~Visioli.
\newblock {\em {Digital Control Engineering: Analysis and Design}}, volume 303
  of {\em Electronics \& Electrical}.
\newblock Elsevier/Academic Press, 2009.

\bibitem{fialho1994stability}
I.~J. Fialho and T.~T. Georgiou.
\newblock On stability and performance of sampled-data systems subject to
  wordlength constraint.
\newblock {\em IEEE Transactions on Automatic Control}, 39(12):2476--2481,
  1994.

\bibitem{Franklin15}
G.~Franklin, D.~Powell, and A.~Emami-Naeini.
\newblock {\em Feedback Control of Dynamic Systems}.
\newblock Pearson, 7th edition, 2015.

\bibitem{FLD+11}
G.~Frehse, C.~L. Guernic, A.~Donz\'e, R.~Ray, O.~Lebeltel, R.~Ripado,
  A.~Girard, T.~Dang, and O.~Maler.
\newblock {SpaceEx}: Scalable verification of hybrid systems.
\newblock In {\em CAV}, volume 6806 of {\em LNCS}, pages 379--395. Springer,
  2011.

\bibitem{horn1990matrix}
R.~A. Horn and C.~Johnson.
\newblock {\em Matrix Analysis}.
\newblock Cambridge University Press, 1990.

\bibitem{itzhaky2010simple}
S.~Itzhaky, S.~Gulwani, N.~Immerman, and M.~Sagiv.
\newblock A simple inductive synthesis methodology and its applications.
\newblock In {\em {OOPSLA}}, pages 36--46. {ACM}, 2010.

\bibitem{DBLP:conf/vmcai/KroeningS03}
D.~Kroening and O.~Strichman.
\newblock Efficient computation of recurrence diameters.
\newblock In {\em Verification, Model Checking, and Abstract Interpretation
  (VMCAI)}, volume 2575 of {\em LNCS}, pages 298--309. Springer, 2003.

\bibitem{gangli1}
G.~Li.
\newblock {On pole and zero sensitivity of linear systems}.
\newblock {\em IEEE Trans. on Circuits and Systems--I: Fundamental Theory and
  Applications}, 44(7):583--590, 1997.

\bibitem{liberzon2003hybrid}
D.~Liberzon.
\newblock Hybrid feedback stabilization of systems with quantized signals.
\newblock {\em Automatica}, 39(9):1543--1554, 2003.

\bibitem{liu16}
J.~Liu and N.~Ozay.
\newblock Finite abstractions with robustness margins for temporal logic-based
  control synthesis.
\newblock {\em Nonlinear Analysis: Hybrid Systems}, 22:1--15, 2016.

\bibitem{mazo2010pessoa}
M.~Mazo, Jr., A.~Davitian, and P.~Tabuada.
\newblock {PESSOA:} {A} tool for embedded controller synthesis.
\newblock In {\em Computer Aided Verification ({CAV})}, volume 6174 of {\em
  LNCS}, pages 566--569. Springer, 2010.

\bibitem{moore1966interval}
R.~E. Moore.
\newblock {\em Interval analysis}, volume~4.
\newblock Prentice-Hall Englewood Cliffs, 1966.

\bibitem{maglev}
V.~A. Oliveira, E.~F. Costa, and J.~B. Vargas.
\newblock Digital implementation of a magnetic suspension control system for
  laboratory experiments.
\newblock {\em IEEE Transactions on Education}, 42(4):315--322, Nov 1999.

\bibitem{oudjida2014design}
A.~K. Oudjida, N.~Chaillet, A.~Liacha, M.~L. Berrandjia, and M.~Hamerlain.
\newblock Design of high-speed and low-power finite-word-length {PID}
  controllers.
\newblock {\em Control Theory and Technology}, 12(1):68--83, 2014.

\bibitem{park2016scalable}
J.~Park, M.~Pajic, I.~Lee, and O.~Sokolsky.
\newblock Scalable verification of linear controller software.
\newblock In {\em Tools and Algorithms for the Construction and Analysis of
  Systems (TACAS)}, LNCS, pages 662--679. Springer, 2016.

\bibitem{picasso2003stabilization}
B.~Picasso and A.~Bicchi.
\newblock Stabilization of {LTI} systems with quantized state-quantized input
  static feedback.
\newblock In {\em 6th International Workshop on Hybrid Systems: Computation and
  Control}, pages 405--416. Springer, 2003.

\bibitem{DBLP:conf/cdc/RavanbakhshS15}
H.~Ravanbakhsh and S.~Sankaranarayanan.
\newblock Counter-example guided synthesis of control {Lyapunov} functions for
  switched systems.
\newblock In {\em Conference on Decision and Control ({CDC})}, pages
  4232--4239, 2015.

\bibitem{DBLP:conf/emsoft/RavanbakhshS16}
H.~Ravanbakhsh and S.~Sankaranarayanan.
\newblock Robust controller synthesis of switched systems using counterexample
  guided framework.
\newblock In {\em {EMSOFT}}, pages 8:1--8:10. {ACM}, 2016.

\bibitem{DBLP:conf/hybrid/RouxJG15}
P.~Roux, R.~Jobredeaux, and P.~Garoche.
\newblock Closed loop analysis of control command software.
\newblock In {\em {HSCC}}, pages 108--117. {ACM}, 2015.

\bibitem{DBLP:conf/asplos/Solar-LezamaTBSS06}
A.~Solar{-}Lezama, L.~Tancau, R.~Bod{\'{\i}}k, S.~A. Seshia, and V.~A.
  Saraswat.
\newblock Combinatorial sketching for finite programs.
\newblock In {\em ASPLOS}, pages 404--415. {ACM}, 2006.

\bibitem{converters}
R.~H.~G. Tan and L.~Y.~H. Hoo.
\newblock {DC}-{DC} converter modeling and simulation using state space
  approach.
\newblock In {\em IEEE Conference on Energy Conversion, {CENCON}}, pages
  42--47, Oct 2015.

\bibitem{DBLP:conf/hybrid/WangGRJF16}
T.~E. Wang, P.~Garoche, P.~Roux, R.~Jobredeaux, and E.~Feron.
\newblock Formal analysis of robustness at model and code level.
\newblock In {\em {HSCC}}, pages 125--134. {ACM}, 2016.

\bibitem{DBLP:journals/automatica/WuLCC09}
J.~Wu, G.~Li, S.~Chen, and J.~Chu.
\newblock Robust finite word length controller design.
\newblock {\em Automatica}, 45(12):2850--2856, 2009.

\bibitem{zamani2014}
M.~Zamani, M.~Mazo, and A.~Abate.
\newblock Finite abstractions of networked control systems.
\newblock In {\em IEEE CDC}, pages 95--100, 2014.

\end{thebibliography}


\appendix

\section{Stability of Closed-loop Models}
\label{sec:appendix-stability}
\subsection{Stability of closed-loop models with fixed-point controller error}
\label{sec:stab_FWL}
The proof of Jury's criterion~\cite{fadali} relies on the fact that the relationship between states and
next states is defined by $x_{k+1} = (A_d - B_dK) x_k$, all computed
at infinite precision.  When we employ a FWL digital controller, the
operation becomes:
\begin{align*}
x_{k+1} &= A_d \cdot x_{k} -(\mathcal{F}_{\langle I_c,F_c \rangle}(K)\cdot\mathcal{F}_{\langle I_c,F_c \rangle}(x_{k})).  \\
x_{k+1} &= (A_d  - B_dK) \cdot x_k + B_dK\delta, 
\end{align*}
where $\delta$ is the maximum error that can be introduced by the FWL
controller in one step, i.e., by reading the states values once and
multiplying by $K$ once.  We derive the closed form expression for $x_n$ as
follows:
\begin{align*}
x_{1} &= (A_d  - B_dK)x_0 + B_dK\delta \\
x_{2} 
 &=(A_d  - B_dK)^2x_0 + (A_d  - B_dK)B_dK\delta + B_dK\delta \\
x_{n} &= (A_d  - B_dK)^nx_0 + (A_d  - B_dK)^{n-1}B_dK\delta + ... + (A_d  - B_dK)^1B_dK \delta + B_dK\delta \\
  &= (A_d - B_dK)^nx_0 + \sum_{i=0}^{i=n-1}(A_d - B_dK)^iB_dk\delta. 
\end{align*}

The definition of asymptotic stability is that the system converges to a
reference signal, in this case we use no reference signal so an
asymptotically stable system will converge to the origin.  We know that the
original system with an infinite-precision controller is stable, because we
have synthesized it to meet Jury's criterion.  Hence, $(A_d - B_dK)^n x_0$ 
must converge to zero.

The power series of matrices
converges~\cite{horn1990matrix} iff the eigenvalues of the matrix are less
than~1 as follows:
%
$\sum_{i=0}^{\infty}T^i  = (I - T)^{-1}$, 
%
where $I$ is the identity matrix and $T$ is a square matrix. Thus, our system will converge to the value 
\begin{align*}
0 + (I - A_d + B_dK)^{-1}B_dk\delta \,. 
\end{align*}
As a result, if the value $(I - A_d + B_dK)^{-1}B_dk\delta$ is within the
safe space, then the synthesized fixed-point controller results in a safe
closed-loop model.  The convergence to a finite value, however, will not
make it asymptomatically stable.

\section{Errors in LTI models} \label{sec:appendix:LTIbackground}

\subsection{Errors due to numerical representation} \label{appendix:numerical_errors}

We have  used $\mathcal{F}_{\langle I,F \rangle}(x)$ denote a real number
$x$ represented in a fixed point domain, with $I$ bits representing the
integer part and $F$ bits representing the decimal part.  The smallest
number that can be represented in this domain is $c_m=2^{-F}$.  The
following approximation errors will arise in mathematical operations and
representation:
\begin{enumerate}

\item {\bf Truncation:} Let $x$ be a real number, and $\mathcal{F}_{\langle
I,F \rangle}(x)$ be the same number represented in a fixed-point domain as
above.  Then $\mathcal{F}_{\langle I,F \rangle}(x) = x-\delta_T$ where the
error $ \delta_T=x\ \%_{c_m}\ \tilde x$, and $\%_{c_m}$ is the modulus
operation performed on the last bit. 
Thus, $\delta_T$ is the truncation error and it will propagate across
operations.
\item {\bf Rounding:} The following errors appear in basic operations.  Let
$c_1, c_2$ and $c_3$ be real numbers, and $\delta_{T1}$ and $\delta_{T2}$ be
the truncation errors caused by representing $c_1$ and $c_2$ in the
fixed-point domain as above.
\begin{enumerate}
\item Addition/Subtraction: these operations only propagate errors coming
from truncation of the operands, namely $\mathcal{F}_{\langle I,F
\rangle}(c_1) \pm \mathcal{F}_{\langle I,F \rangle}(c_2) = c_3 + \delta_3$
with $|\delta_3| \leq |\delta_{T1}| + |\delta_{T2}|$.
\item Multiplication: $\mathcal{F}_{\langle I,F \rangle}(c_1) \cdot
\mathcal{F}_{\langle I,F \rangle}(c_2) =  c_3 + \delta_3$ with $|\delta_3|
\leq |\delta_{T1}\cdot\mathcal{F}_{\langle I,F \rangle}(c_2)|\allowbreak +
|\delta_{T2}\cdot\mathcal{F}_{\langle I,F \rangle}(c_1)| + c_m$, where
$c_m=2^{-F}$ as above.
\item Division: the operations performed by our controllers in the FWL
domain do not include division.  However, we do use division in computations
at the precision of the plant.  Here the error depends on whether the
divisor is greater or smaller than the dividend:  $\mathcal{F}_{\langle I,F
\rangle}(c_1) / \mathcal{F}_{\langle I,F \rangle}(c_2) = c_3 + \delta_{T3}$
where $\delta_{T3}$ is $(\delta_{T2}\cdot c_1 - \delta_{T1}\cdot
c_2)/(\delta_{T2}^2 - \delta_{T2} c_2)$,
%
\end{enumerate}

\item {\bf Overflow:}
The maximum size of a real number $x$ that can be represented in a fixed
point domain as $\mathcal{F}_{\langle I,F \rangle}(x)$ is $\pm
(2^{I-1}+1-2^{-F})$.  Numbers outside this range cannot be represented by
the domain.  We check that overflow does not occur.

\end{enumerate}

\subsection{Modeling quantization as noise} \label{appendix:quantization-noise}

During any given ADC conversion, the continuous signal will be sampled in
the real domain and transformed by $\mathcal{F}\langle I_{c},F_{c} \rangle
(x)$ (assuming the ADC discretization is the same as the digital
implementation).  This sampling uses a threshold which is defined by the
less significant bit ($q_{c}=c_{m_c}=2^{-F_c}$) 
of the ADC and some non-linearities of the circuitry.  The overall conversion is
$$\mathcal{F}\langle I_{c},F_{c} \rangle(y(t)) = y_k :
y_k \in \left[y(t)-\frac{q_{c}}{2}\ \ \ \ y(t)+\frac{q_{c}}{2}\right] \,.$$
If we denote the error in the conversion by $\nu_k=y_k-y(t)$ where $t = nk$,
and $n$ is the sampling time and $k$ the number of steps, then we may define
some bounds for it $\nu_k \in [-\frac{q_{c}}{2}\ \ \frac{q_{c}}{2}]$.

We will assume, for the purposes of this analysis, that the domain of the
ADC is that of the digital controller (i.e, the quantizer includes any
digital gain added in the code).  The process of quantization in the DAC is
similar except that it is calculating $\mathcal{F}\langle I_{dac},F_{dac}
\rangle (\mathcal{F}\langle I_{c},F_{c} \rangle (x)) $.  If these domains
are the same ($I_{c}=I_{dac},\allowbreak F_{c}=F_{dac}$), or if the DAC
resolution in higher than the ADCs, then the DAC quantization error is equal
to zero.  From the above equations we can now define the ADC and DAC
quantization noises ${\nu_1}_k \in [-\frac{q_1}{2}\ \ \frac{q_1}{2}]$ and
${\nu_2}_k \in [-\frac{q_2}{2}\ \ \frac{q_2}{2}]$, where $q_1=q_{c}$ and 
$q_2=q_\mathit{dac}$.  This is illustrated in
Fig.~\ref{ssec:statefeedbackcontrol} where $Q_1$ is the quantizer of the ADC
and $Q_2$ the quantizer for the DAC.  These bounds hold irrespective of
whether the noise is correlated, hence we may use them to over-approximate
the noise effect on the state space progression over time.  The
resulting dynamics are
\begin{align*}
{x}_{k+1} = {A}_d{x}_k+{B}_d({u}_k+{{\nu}_2}_k), \quad u_k = -K{x}_{k}+{{\nu}_1}_k, 
\end{align*}
which result in the following closed-loop dynamics:
\begin{align*}
{x}_{k+1} &= ({A}_d-{B}_d{K}_d) {x}_k+{B}_d{{\nu}_2}_k +{{\nu}_1}_k \,. 
\end{align*}

\end{document}